\documentclass[11pt]{article}

\usepackage{amsmath,amsthm,amssymb,amsxtra,amsfonts,amsbsy,mathrsfs}
\usepackage[all,2cell]{xy} \UseAllTwocells
\SilentMatrices
\newtheorem{theorem}[subsection]{Theorem}

\newtheorem{lemma}[subsection]{Lemma}
\newtheorem{proposition}[subsection]{Proposition}

\theoremstyle{definition}
\newtheorem{notation-convention}[subsection]{Notations
and Conventions}

\newtheorem{remark}[subsection]{Remark}

\begin{document}
\title{A Note on Nahm's Conjecture in Rank 2 Case}
\author{An Huang\footnote{anhuang@berkeley.edu}, Chul-hee Lee\footnote{chlee@math.berkeley.edu}\\
Department of Mathematics\\
University of California, Berkeley\\
CA 94720-3840 USA\\}
\date{August 2010}
\maketitle
\begin{abstract}
\noindent The aim of this paper is to get a complete list of positive definite symmetric matrices with integer entries $\begin{bmatrix}a&b\\b&d\end{bmatrix}$ such that all complex solutions to the system of equations
\begin{align}
1-x_1=x_1^ax_2^b\nonumber \\ 
1-x_2=x_1^bx_2^d\nonumber
\end{align}
are real. This result is related to Nahm's conjecture in rank 2 case. 
\end{abstract}

\section{Introduction}
The investigation of torus partition functions of certain conformal field theories is one of the central topics relating physics and number theory. From physical considerations, in general, one expects those functions to have nice modular transformation properties. \cite{Zhu} is a fundamental paper devoted to establish such expectations rigorously by using the theory of vertex operator algebras. In \cite{Nahm}, Nahm considered certain rational conformal field theories with integrable perturbations, and argued that their partition functions have a canonical sum representation in terms of $q$-hypergeometric series. Consequently, he conjectured a partial answer to the question of when a particular $q$-hypergeometric series is modular. This is called Nahm's conjecture, and is discussed in \cite{Zagier}. (In particular, one may consult pages 40 and 41 of this paper for a precise statement of the conjecture.) 

To state Nahm's conjecture, we consider the $q$-hypergeometric series
\begin{equation*}
f_{A,B,C}(z)=\sum_{n=(n_1,...,n_r)\in (\mathbb{Z}_{\geq 0})^r}\frac{q^{\frac{1}{2}n^tAn+B^tn+C}}{(q)_{n_1}...(q)_{n_r}},
\end{equation*}
where $(q)_n$ denotes the product $(1-q)(1-q^2)...(1-q^n)$, $A$ is a positive definite symmetric $r\times r$ matrix, $B$ is a vector of length $r$, and $C$ is a scalar, all three with rational coefficients. We will call $r$ the rank throughout the paper. Associated with $A=(a_{ij})$ we consider the system of $r$ equations of $r$ variables $x_1,...,x_r$
\begin{equation} \label{N1}
1-x_i=\prod_{j=1}^{r}x_j^{a_{ij}}  \hspace{2.2cm}                (i=1,...,r).
\end{equation}
The definition of \eqref{N1} needs to be made more precise when there are nonintegral entries in $A$. However, in this paper we only consider the case when all entries are integral, so there is no problem. There are only finitely many solutions to \eqref{N1}, so all solutions lie in $\overline{\mathbb{Q}}$. For any solution $x=(x_1,...,x_r)$, we consider the element $\xi_x=[x_1]+...+[x_r]\in \mathbb{Z}[F]$, where $F$ is the number field $\mathbb{Q}(x_1,...,x_r)$. $\xi_x$ defines an element in the Bloch group $\mathcal{B}(F)$ by a standard construction \cite{Zagier}. Then Nahm's conjecture asserts that the following are equivalent:
\begin{enumerate}
\item [(i)] The element $\xi_x$ is a torsion element of $\mathcal{B}(F)$ for every solution $x$ of \eqref{N1}
\item [(ii)] There exist $B$ and $C$ such that $f_{A,B,C}(z)$ is a modular function.
\end{enumerate}

It is not hard to show that Nahm's conjecture holds in rank 1 case (i.e., $r=1$): one may see \cite{Zagier} for example. In particular, $A=1, 2$ are the only integral values of $A$ that satisfy (i). However, from rank 2 and above, both directions are open. Obviously, it will be very useful if one can get a complete list of matrices $A$ such that condition (i) above holds (Let us denote this list by $L$). Presumably, however, this is a hard task. Instead, one may ask an easier question: to determine the set of matrices $A$ such that all solutions to \eqref{N1} are real. For any number field $F$, it is well known that the free part of $\mathcal{B}(F)$ is isomorphic to $\mathbb{Z}^{r_2}$, where $2r_2$ is the number of complex embeddings of $F$. So if all solutions to \eqref{N1} are real, then $F$ is totally real, $r_2=0$, and any element of $\mathcal{B}(F)$ is torsion. Therefore for each rank $r$, the set of these matrices is a subset of $L$. We wish to investigate this subset for two reasons: first, it is reasonable to expect that this subset constitutes a substantial part of $L$ (From \cite{Zagier} and \cite{Nahm}, one may look at the available examples of Nahm's conjecture to get an idea of this. Except some trivial infinite families, most examples are in this subset.); second, we hope this subset is much more tractable. 

In this paper we focus our attention to rank 2 case, and consider only the case when all entries of $A$ are integers. With these restrictions, we will identify this subset exactly by using Bezout's theorem. The idea is very simple: when $r=2$, we introduce a new variable $z$ and consider homogeneous equations corresponding to \eqref{N1}. Then Bezout's theorem tells us the exact number of complex solutions to the system of homogeneous equations, counting multiplicity. There are solutions to the homogeneous system of equations, corresponding to $z=0$, which are not solutions to \eqref{N1}. We can estimate the multiplicities of these  solutions, and this gives us a lower bound for the number of solutions to \eqref{N1}, counting multiplicity. The technique we will employ to estimate the multiplicities is the method of local analytic parametrization. For reference, one can see \cite{Walker}, chapter IV, sections 1-5. On the other hand, we prove an absolute upper bound for the number of real solutions to \eqref{N1}, counting multiplicity. Combining these, together with the condition that all solutions to \eqref{N1} are real, we obtain inequalities for $a,b,d$, which are sharp enough to enable us to determine all possibilities. 

Our aim is to prove
\begin{theorem}
If $\begin{bmatrix}a&b\\b&d\end{bmatrix}$ is a positive definite symmetric matrix with integer entries such that all complex solutions to the system of equations
\begin{equation}\label{N}
\begin{aligned}
1-x_1=x_1^ax_2^b\\
1-x_2=x_1^bx_2^d
\end{aligned}
\end{equation}
are real, then $\begin{bmatrix}a&b\\b&d\end{bmatrix}$ equals one of $\begin{bmatrix}2&1\\1&1\end{bmatrix}$, $\begin{bmatrix}1&1\\1&2\end{bmatrix}$,    $\begin{bmatrix}4&2\\2&2\end{bmatrix}$, $\begin{bmatrix}2&2\\2&4\end{bmatrix}$, $\begin{bmatrix}1&-1\\-1&2\end{bmatrix}$, $\begin{bmatrix}2&-1\\-1&1\end{bmatrix}$, $\begin{bmatrix}2&-1\\-1&2\end{bmatrix}$, $\begin{bmatrix}2&0\\0&2\end{bmatrix}$, $\begin{bmatrix}1&0\\0&2\end{bmatrix}$, $\begin{bmatrix}2&0\\0&1\end{bmatrix}$, or $\begin{bmatrix}1&0\\0&1\end{bmatrix}$.
\end{theorem}

In section 2, we will assume $b>0$ and prove theorem 1.1 in this case. After we finish the proof of $b>0$ case, it should be very clear how to generalize the proof to the case when $b$ is negative ($b=0$ case is trivial as it reduces to the rank 1 case). We will include the proof for the latter case in section 3. At the end of this paper, we will also indicate a possible generalization of the above method to higher rank cases.
\begin{remark}
Concerning the list of rank 2 examples in \cite{Zagier}, except the trivial infinite family corresponding to central charge one representations of the Virasoro algebra, there is essentially only one known example from integral rank 2 case, where all solutions give rise to torsion, but not all solutions are real. Namely, the matrix $\begin{bmatrix}4&1\\1&1\end{bmatrix}$ (or $\begin{bmatrix}1&1\\1&4\end{bmatrix}$).  
\end{remark}
\section{The case when $b>0$}
Without loss of generality, we assume $a\geq d$ in this section. Thus $a>b\geq 1$.
We consider the following homogeneous system of equations in $\mathbb{CP}^2$:
\begin{equation}\label{H}
\begin{aligned}
z^{a+b}-x_1z^{a+b-1}-x_1^ax_2^b=0\\
z^{b+d}-x_2z^{b+d-1}-x_1^bx_2^d=0
\end{aligned}
\end{equation}  
By Bezout's theorem, the number of solutions to \eqref{H}, counting multiplicity, is equal to $(b+a)(b+d)$. Obviously, complex solutions to \eqref{H} with $z\neq 0$ are in one-to-one correspondence with complex solutions to \eqref{N}. Moreover, \eqref{H} has two solutions with $z=0$: $[x_1:x_2:z]=[1:0:0]$, and $[x_1:x_2:z]=[0:1:0]$. Let us denote their multiplicities by $i_1$ and $i_2$, respectively. And we denote the determinant of the matrix $A$ by $\Delta$. We have the following lemmas for $(b+a)(b+d)-i_1-i_2$.
\begin{lemma}
Let $g$ denote the greatest common divisor of $2b$ and $d$.
If $\Delta=d$ and $b>d$, we have $i_1\leq b(b+d)+g$.\\
Otherwise, $i_1=\text{min}\left\{b(b+d),d(a+b-1)\right\}$.
\end{lemma}
\begin{proof}
Take $x_1=1$, then \eqref{H} are reduced to
\begin{subequations} 
\begin{align}
z^{a+b}-z^{a+b-1}-x_2^b=0\label{1}\\
z^{b+d}-x_2z^{b+d-1}-x_2^d=0\label{2}
\end{align} 
\end{subequations}

We would like to take a local analytic parametrization of \eqref{1} at $z=x_2=0$, and compute $i_1$ with that. Let $e$ denote the greatest common divisor of $a+b-1$ and $b$, and write $a+b-1=eu_1$, $b=eu_2$. Take 
\begin{equation}\label{z}
z=t^{u_2},
\end{equation} 
then we have $(\frac{x_2}{t^{u_1}})^{eu_2}=t^{u_2}-1$. We use $(1-t^{u_2})^{\frac{1}{eu_2}}$ to denote its Taylor series, then
\begin{equation}\label{x}
x_2=\frac{t^{u_1}(1-t^{u_2})^{\frac{1}{eu_2}}\omega}{\omega'},
\end{equation}
where $\omega$ is a $b$th root of unity, and $\omega'$ is any chosen primitive $2b$th root of unity.

It is straightforward to check that equations \eqref{z} and \eqref{x} give a local analytic parametrization of \eqref{1} at $z=x_2=0$, provided that $\omega$ varies among $\omega_1^k$, where $\omega_1$ is a primitive $b$th root of unity, and $k=1,2,...,e$. The point is that for each point on the affine curve \eqref{1}, there exists a unique pair of $k$ and $t$, such that \eqref{z} and \eqref{x} give rise to the coordinates of that point. For definiteness, let us choose $\omega_1=(\omega')^2$.(Note that these choices are not unique. However, any choice will give rise to the same answer for $i_1$, of course.)

Substituting \eqref{z} and \eqref{x} into \eqref{2}, then \eqref{2} becomes
\begin{equation}\label{rr}
t^{(b+d)u_2}-\frac{t^{u_1+(b+d-1)u_2}(1-t^{u_2})^{\frac{1}{eu_2}}\omega}{\omega'}-t^{du_1}(\frac{(1-t^{u_2})^{\frac{1}{eu_2}}\omega}{\omega'})^d.
\end{equation}
By the theory of local analytic parametrization, $i_1$ equals the sum over $k$ of the degrees of the lowest degree terms in $t$ in \eqref{rr}.

Since $a>1$, $u_1>u_2$, $(b+d)u_2<u_1+(b+d-1)u_2$.

If $\Delta\neq d$, then $(b+d)u_2\neq du_1$. Therefore, for each choice of $\omega$, the degree of the lowest degree term in \eqref{rr} equals $\text{min}\left\{(b+d)u_2,du_1\right\}$. So we have $i_1= \text{min}\left\{b(b+d),d(a+b-1)\right\}$.

If $\Delta=d$, then $(b+d)u_2=du_1$. Since $d=\Delta=ad-b^2$, $b^2=d(a-1)$. So $a-1\geq b\geq d$. There are two possibilities:

(i) $b=d$, then $a-1=b$. We have $i_1=I(z^{2b+1}-z^{2b}-x_2^b,z^{2b}-x_2z^{2b-1}-x_2^b)=I(z^{2b+1}-z^{2b}-x_2^b,z^{2b}-x_2z^{2b-1}-(z^{2b+1}-z^{2b}))=b(2b)=b(b+d)$,
as the intersection multiplicity equals the product of multiplicities of the point on each curve, if the two curves share no common tangent lines at the point.

(ii) $b>d$. Then $a-1>b$. Therefore $u_1+(b+d-1)u_2$ is greater than $du_1+u_2$, which is the degree of the second lowest degree term of $t^{du_1}(\frac{(1-t^{u_2})^{\frac{1}{eu_2}}\omega}{\omega'})^d$. Consequently, the degree of the lowest degree term in \eqref{rr} equals $(b+d)u_2$ if $(\frac{\omega}{\omega'})^d\neq 1$, and equals $du_1+u_2$ if $(\frac{\omega}{\omega'})^d=1$. We write $d=gd'$, $2b=gb'$. Since $g$ is the greatest common divisor of $2b$ and $d$, $d'$ and $b'$ are coprime. Therefore, for $\omega=\omega_1^k$ and $\omega_1=(\omega')^2$, $(\frac{\omega}{\omega'})^d=1$ only if $b'$ divides $2k-1$, which can happen for at most $\frac{2e}{b'}$ many $k$. We then get an estimate for $i_1$:
\begin{equation}
i_1\leq (b+d)u_2(e-\frac{2e}{b'})+(du_1+u_2)(\frac{2e}{b'})=b(b+d)+g.
\end{equation}  
\end{proof}
\begin{lemma}
If $\Delta=a$, then $i_2\leq b(a+b)+d-1$.\\
Otherwise, $i_2=\text{min}\left\{b(b+a),a(d+b-1)\right\}$.
\end{lemma}
\begin{proof}
The proof is very similar to the proof of the above lemma. Everywhere one replaces $a$,$e$,$u_1$,$u_2$ by $d$,$f$,$v_1$,$v_2$, respectively. If $\Delta\neq a$ and $d>1$, then $v_1>v_2$, and the same estimate gives $i_2=\text{min}\left\{b(b+a),a(d+b-1)\right\}$. If $\Delta=a$, then $d>1$, and the difference with the previous case is that now we have $v_1+v_2(a+b-1)<av_1+v_2$, so the terms of degree $v_1+v_2(a+b-1)$ survives, and a rude estimate gives $i_2\leq f(v_1+v_2(a+b-1))=b(a+b)+d-1$. If $d=1$, then $a>\Delta$, $(a+b)v_2=v_1+(a+b-1)v_2>av_1$, and $i_2=afv_1=a(b+d-1)$.
\end{proof}
Having both lemmas, now let us estimate $(a+b)(d+b)-i_1-i_2$. Obviously $\Delta=a$ and $\Delta=d$ cannot both happen. We have the following result.
\begin{lemma}
If $\Delta=a$ or $\Delta=d<b$, $(a+b)(d+b)-i_1-i_2\geq a-d$.\\
Otherwise, $(a+b)(d+b)-i_1-i_2\geq a$.
\end{lemma}
\begin{proof}
If $\Delta\neq a$ and one of $\Delta\neq d$, $\Delta=d=b$ holds, then 
\begin{equation}
\begin{aligned}
&(a+b)(d+b)-i_1-i_2\\
&= (a+b)(d+b)-\text{min}\left\{b(b+d),d(a+b-1)\right\}-\text{min}\left\{b(b+a),a(d+b-1)\right\}.
\end{aligned}
\end{equation}
Therefore, 
\begin{equation}
(a+b)(d+b)-i_1-i_2=\Delta+\text{max}\left\{0,d-\Delta\right\}+\text{max}\left\{0,a-\Delta\right\}\geq a.
\end{equation}
If $\Delta=a$, then $(a+b)(d+b)-i_1-i_2\geq (a+b)(d+b)-b(b+d)-(b(b+a)+(d-1))=a-d+1$.\\
If $\Delta=d<b$, then $(a+b)(d+b)-i_1-i_2\geq (a+b)(d+b)-(b(b+d)+g)-a(d+b-1)=a-g\geq a-d$.
\end{proof}     
Next we prove an absolute upper bound for the number of real solutions to \eqref{N} (counting multiplicity).
\begin{lemma}
The number of real solutions to \eqref{N}, counting multiplicity, is at most 9.
\end{lemma}
\begin{proof}
First of all, there is exactly one solution in $(0,1)^2$ with multiplicity one. (In \cite{Zagier}, Zagier already mentioned that there is exactly one solution in this domain, and moreover this holds in much more general case.) By eliminating $x_2$, we get
\begin{equation}\label{3}
(\frac{1-x_1}{x_1^a})^{\frac{1}{b}}+(\frac{1-x_1}{x_1^{\frac{\Delta}{d}}})^{\frac{d}{b}}-1=0.
\end{equation}

In the domain $(0,1)^2$, there is no ambiguity on the definition of \eqref{3}, and obviously solutions to \eqref{N} in the domain are in one-to-one correspondence with solutions to \eqref{3} in $(0,1)$. Moreover, for any solution to \eqref{3} given by $x_1=\tau$, we have the well-defined multiplicity as the valuation of the left hand side of \eqref{3} at the point $<x_1-\tau>$ on the affine $\mathbb{A}_{\mathbb{C}}^1$. In other words, expand the left hand side of \eqref{3} as a formal power series in $(x_1-\tau)$, the multiplicity of the solution $x_1=\tau$ equals the degree of the nonzero lowest degree term. Furthermore, this multiplicity is the same as the multiplicity of the corresponding solution to \eqref{N}.

However, since the left hand side of \eqref{3} is strictly decreasing in $(0,1)$, one easily sees that there is exactly one solution in $(0,1)$. Since the derivative must be negative, the multiplicity is one.

Next, we consider solutions outside $(0,1)^2$. 

If $b$ is odd, then \eqref{3} is defined unambiguously if we concern only real solutions, since there is only one branch of the function $x\rightarrow x^{\frac{1}{b}}$ which maps real numbers to real numbers.

If we have at least one solution with $x_1>1$ , then $d$ must be even, and $x_2<0$. This is case (I). We will show that there are at most 2 solutions in this case, counting multiplicity.

If there exists at least one solution with $x_1<0$, then there are two possibilities:

$x_2>0$. In this case we have $x_2>1$, and $a$ must be even. Exchanging indices 1 and 2, we see that from case (I) there can be at most two solutions for $x_2$, so at most two solutions for $x_1$ as well.

$x_2<0$. In this case we must have both a and d to be odd. This is case (II). We will show that there are at most 6 solutions in this case, counting multiplicity.

Note that any solution to $\eqref{N}$ with $x_1\in (0,1)$ has to be the unique solution in $(0,1)^2$, so the above exhausted all possibilities of real solutions.

If $b$ is even, then we have three possibilities:\\
$x_1<0$, and $x_2<0$. We call this case (III). In this case both $a$ and $d$ must be even, and we will show that there are at most 6 solutions, counting multiplicity.\\
$x_1<0$, and $0<x_2<1$. So $a$ has to be even. We call this case (IV), and we will show that there is at most 1 solution, counting multiplicity.\\
$x_2<0$, and $0<x_1<1$. So $d$ has to be even. This is the same as case (IV) with $x_1$ and $x_2$ switched, so there is at most 1 solution, counting multiplicity.

For case (I), \eqref{3} can be rewritten as
\begin{equation}
f(x_1)=-(\frac{x_1-1}{x_1^a})^{\frac{1}{b}}+(\frac{x_1-1}{x_1^{\frac{\Delta}{d}}})^{\frac{d}{b}}-1=0.
\end{equation}
It is easy to see that we need to have $\Delta<d$, in order that this equation has a solution for $x_1>1$.  We have the following identity for the derivative function:
\begin{equation} 
b(x_1-1)x_1f'(x_1)=x_1^bs^d[(d-\Delta)x_1+\Delta]-s[a-(a-1)x_1],
\end{equation}
where 
\begin{equation} 
s=(\frac{x_1-1}{x_1^a})^{\frac{1}{b}}.
\end{equation}
(Note that multiplying by invertible elements such as $x_1-1$, or $x_1$ in local rings do not affect the multiplicities of solutions. And we are making use of the fact that if a polynomial equation has $N$ real solutions in an open interval, counting multiplicity, then its derivative has $N-1$ real solutions in the same interval, counting multiplicity.)

So $f'(x_1)$ can possibly have a solution only when $x_1<\frac{a}{a-1}$. In this case, when $0<x_1<\frac{a}{a-1}$, $f'(x_1)=0$ iff 
\begin{equation}\label{M1}
s^{d-1}x_1^b[(d-\Delta)x_1+\Delta]=a-(a-1)x_1. 
\end{equation}
But it is easy to see that the left hand side of \eqref{M1} is strictly increasing, and the right hand side is strictly decreasing. So $f'(x_1)=0$ has at most 1 solution counting multiplicity. Consequently, $f(x)=0$ has at most 2 solutions, counting multiplicity.

For case (IV), by eliminating $x_2$ we have
\begin{equation}\label{M4}
(\frac{1+x}{x^a})^{\frac{1}{b}}+(\frac{1+x}{x^{\frac{\Delta}{d}}})^{\frac{d}{b}}-1=0,
\end{equation} 
where $x_1=-x$, $x>1$.

If $\Delta\leq d$, then the left hand side of \eqref{M4} is always greater than 0, so we don't have solutions. If $\Delta>d$, then the left hand side of \eqref{M4} is strictly decreasing, so we have at most one real solution counting multiplicity.

For case (II) and (III), denote $x_1=-x$, with $x>0$, by eliminating $x_2$ we have
\begin{equation}
g(x)=-(\frac{1+x}{x^a})^{\frac{1}{b}}+(\frac{1+x}{x^{\frac{\Delta}{d}}})^{\frac{d}{b}}-1=0.
\end{equation}
For the derivative function, we have
\begin{equation} \label{M23}
b(1+x)tg'(x)=-x^bs_1^d[(\Delta-d)x+\Delta]+s_1[a+(a-1)x],
\end{equation}
where 
\begin{equation} 
s_1=(\frac{1+x}{x^a})^{\frac{1}{b}}.
\end{equation}
So real solutions to $g'(x)=0$ are equivalent to real solutions to (counting multiplicity)
\begin{equation}\label{F}
(\Delta-d)\frac{(1+x)^{u+1}}{x^v}+d\frac{(1+x)^u}{x^v}-(a+(a-1)x)=0,
\end{equation}
where $u=\frac{d-1}{b}$, $v=\frac{\Delta-a}{b}$. But we have
\begin{equation*} 
(\frac{(1+x)^u}{x^v})''=\frac{(1+x)^{u-2}}{x^{v+2}}\times \text{(quadratic polynomial of $x$)}.
\end{equation*}
So the second order derivative of the left hand side of \eqref{F} equals 0 iff 
\begin{equation*}
\frac{(1+x)^{u-2}}{x^{v+2}}\times \text{(a polynomial of $x$ of degree at most 3)}=0.
\end{equation*}

So the second order derivative of the left hand side of \eqref{F} equals 0 has at most 3 solutions, $g'(x)=0$ has at most 5 solutions, and $g(x)=0$ has at most 6 solutions, all counting multiplicity.
Combining all the above, and enumerate all 8 cases of $a,b,d$ being even or odd as in table 1, the lemma is proved.
\begin{table}
\caption{upper bound for the number of real solutions $\#$ (for $a,b,d$, 1 denotes odd, and 0 denotes even)}
\begin{center} \begin{tabular}{| c | c | c | c | c | c | c | c | c |}
\hline $b$&1&1&1&1&0&0&0&0\\
\hline $a$&1&0&0&1&1&1&0&0\\
\hline $d$&0&0&1&1&1&0&1&0\\
\hline $\#\leq$& 1+2&1+2+2&1+2&1+6&1&1+1&1+1&1+1+1+6\\
\hline \end{tabular} \end{center}
\end{table}
\end{proof}
Now let us prove theorem 1.1 for the case $b>0$.
\begin{proof}
In \cite{Zagier}, they first searched for matrices $A$ with integral entries (actually they did it for matrices with rational entries with certain bounds on numerator and denominator) whose absolute values are less than or equal to $100$ such that $L(\xi)\in \pi^2\mathbb{Q}$, where $\xi=\left[x_1\right]+\left[x_2\right]$ is the Bloch group element corresponding to the unique solution $x_1, x_2$ to \eqref{N} in $\left(0,1\right)^2$, and $L(\xi)=L(x_1)+L(x_2)$, where $L(x)$ is the Rogers dilogarithm function. Among these they then identified all matrices such that condition (i) of Nahm's conjecture holds. On the other hand, as we have explained in the introduction, for any matrix $A$ satisfying the condition of theorem 1.1, $A$ satisfies condition (i) of Nahm's conjecture, and in particular $L(\xi)\in \pi^2\mathbb{Q}$, which follows from the well-definedness of the regulator map (for details one may see \cite{Zagier}). Meanwhile, combining lemmas 2.3 and 2.4, we have $a\leq 9$, or $\Delta=a$, $a-b\leq 9$, or $\Delta=d<b$, $a-b\leq 9$. From each of these conditions one easily derives that $a,b,d<100$, so $A$ has to be in the Zagier's list. Therefore, only the matrices $\begin{bmatrix}2&1\\1&1\end{bmatrix}$ and $\begin{bmatrix}4&2\\2&2\end{bmatrix}$ survive.
\end{proof}
\section{The case when $b$ is negative}
We will prove the following
\begin{proposition}
if $\begin{bmatrix}a&b\\b&d\end{bmatrix}$ is a positive definite symmetric matrix with integer entries such that all complex solutions to the system of equations
\begin{align}
1-x_1=x_1^ax_2^b\nonumber \\ 
1-x_2=x_1^bx_2^d\nonumber
\end{align}
are real, and $b$ is negative, then $a,-b,d\leq 20$.
\end{proposition}
This proposition will be a consequence of the following two lemmas 3.2 and 3.3, and it immediately implies the full theorem 1.1, again as there are only finitely many remaining cases to check to identify all matrices among these such that all solutions to \eqref{N} are real, and all these cases are already checked in \cite{Zagier}. The proof will be brief since it is similar to the proof of the case $b>0$ above.

Let us write $b=-c$, with $c>0$. Without loss of generality, we assume $a\geq d$. Then $a\geq c+1$. We consider the system of equations
\begin{equation}\label{Z}
\begin{aligned}
x_2^c(1-x_1)=x_1^a\\
x_1^c(1-x_2)=x_2^d
\end{aligned}
\end{equation}
Except for $x_1=x_2=0$, solutions to \eqref{Z} are in one-to-one correspondence with solutions to \eqref{N}. Using the method of local analytic parametrization, one easily sees that the multiplicity of this solution is $c^2$: write $a=h_1a_1$, $c=h_1c_1$, where $h_1$ is the greatest common divisor of $a$ and $c$. Then $x_1=t^{c_1}$, $x_2=\frac{t^{a_1}\theta^k}{(1-t^{c_1})^{\frac{1}{h_1c_1}}}$, where $\theta$ is a primitive $c$th root of unity, and $k$ varies among $1,2,...,h_1$ defines a local analytic parametrization of $x_1^a-x_2^c(1-x_1)$ at $x_1=x_2=0$. Substituting this parametrization into $x_1^c(1-x_2)=x_2^d$ while keeping in mind that $ad>c^2$, one sees easily that for every $k$, the degree of the lowest degree term is $c_1c$. So the multiplicity of this solution is $c^2$. 
\begin{lemma}
The number of solutions to \eqref{Z} other than $x_1=x_2=0$, counting multiplicity, is at least $a-1$.
\end{lemma}
\begin{proof}
Let us denote the quantity in the above lemma by $n$. We introduce a new variable $z$, and consider homogeneous equations corresponding to \eqref{Z}. Since $a\geq c+1$, we need to discuss two cases:\\
Case (i): $d\geq c+1$

In this case, the homogeneous equations are
\begin{equation}
\begin{aligned}\label{W}
x_2^c(z^{a-c}-x_1z^{a-c-1})=x_1^a\\
x_1^c(z^{d-c}-x_2z^{d-c-1})=x_2^d
\end{aligned}
\end{equation}

If $a>c+1$, then $z=0$ implies $x_1=x_2=0$. So we don't have solutions to \eqref{W} with $z=0$. Thus, Bezout's theorem implies that $n\geq ad-c^2=\Delta>a-1$.

If $a=c+1$, then $d=c+1$, \eqref{W} becomes
\begin{align}
x_2^c(z-x_1)=x_1^{c+1}\nonumber\\
x_1^c(z-x_2)=x_2^{c+1}\nonumber
\end{align}
There are $c$ distinct solutions with $z=0$: $[x_1:x_2:z]=[1:w_k:0]$, where $w_k=e^{\frac{\pi i+2k\pi i}{c}}$, $k=0,1,...,c-1$. We calculate the local valuation of the function $x_2^cz-x_2^c-1$ on the curve $z-x_2-x_2^{c+1}=0$. In fact we have $x_2^cz-x_2^c-1=-(x_2+1)(x_2^c+1)$. So $V(x_2^cz-x_2^c-1)$ equals $2$ if $w_k=-1$, and $1$ otherwise. There is at most one $k$ such that $w_k=-1$, so the sum of multiplicities of these solutions is at most $c+1$. Therefore $n\geq (c+1)^2-(c+1)-c^2=c=a-1$.\\
Case (ii): $d\leq c$

In this case, the homogeneous equations are
\begin{equation}
\begin{aligned}\label{c2}
x_2^c(z^{a-c}-x_1z^{a-c-1})=x_1^a\\
x_1^c(z-x_2)=x_2^dz^{c+1-d}
\end{aligned}
\end{equation}
We have only one solution with $z=0$, namely $[x_1:x_2:z]=[0:1:0]$. Let us denote the multiplicity of this solution by $i$.

Again we use the method of local analytic parametrization. Let $x_2=1$, \eqref{c2} are reduced to
\begin{subequations}
\begin{align}
z^{a-c}-x_1z^{a-c-1}-x_1^a=0\label{cc2}\\
x_1^c(z-1)-z^{c+1-d}=0 \label{cc3}
\end{align}
\end{subequations}
Write $c=h_2c_2$, $c+1-d=h_2d_2$, where $h_2$ is the greatest common divisor of $c$ and $c+1-d$. Then $z=t^{c_2}$, $x_1=\frac{t^{d_2}\theta^k}{(t^{c_2}-1)^{\frac{1}{h_2c_2}}}$, where $k$ varies among $1,2,...,h_2$ defines a local analytic parametrization of \eqref{cc3} at $x_1=x_2=0$. Substituting this parametrization into \eqref{cc2}, we get
\begin{equation}
t^{c_2(a-c)}-\frac{t^{d_2+c_2(a-c-1)}\theta^k}{(t^{c_2}-1)^{\frac{1}{h_2c_2}}}-\frac{t^{ad_2}\theta^{ak}}{(t^{c_2}-1)^{\frac{a}{h_2c_2}}}.
\end{equation}
   
If $d>1$, then $c_2(a-c)>d_2+c_2(a-c-1)$, and we have two possibilities:\\
(i) $a+d\neq \Delta+1$, then $d_2+c_2(a-c-1)\neq ad_2$, and therefore $i\leq c+1-d+c(a-c-1)=ac-c^2-d+1<c(a-c)$,\\
(ii) $a+d=\Delta+1$, then $d_2+c_2(a-c-1)=ad_2>c_2(a-c)-c_2$, so the term $t^{c_2(a-c)}$ survives, and $i\leq c(a-c)$.

Either case, we have $n\geq a(c+1)-c^2-c(a-c)=a$.

If $d=1$, the second lowest degree term in $\frac{t^{d_2+c_2(a-c-1)}\theta^k}{(t^{c_2}-1)^{\frac{1}{h_2c_2}}}$ survives, and we have $i\leq a-c+1\leq a$, and $n\geq a(c+1)-c^2-a=c(a-c)\geq a-1$.
\end{proof}
\begin{lemma}
The number of real solutions to \eqref{Z} other than $x_1=x_2=0$, counting multiplicity, is at most $19$.
\end{lemma}
\begin{proof}
Except for the solution in $(0,1)^2$, again we have four possible cases, identical to the case when $b>0$. In each of these cases, the same argument as in the cases (II) and (III) for $b$ being positive works, which produces an upper bound $6$ for the number of real solutions in each case. Again by enumerating all $8$ cases of $a,c,d$ being even or odd as in table 2, we get an upper bound for the number of solutions counting multiplicity to be $1+6+6+6=19$:
\begin{table}
\caption{upper bound for the number of real solutions $\#$ (for $a,c,d$, 1 denotes odd, and 0 denotes even)}
\begin{center} \begin{tabular}{| c | c | c | c | c | c | c | c | c |}
\hline $c$&1&1&1&1&0&0&0&0\\
\hline $a$&1&0&0&1&1&1&0&0\\
\hline $d$&0&0&1&1&1&0&1&0\\
\hline $\#\leq$& 1+6&1+6+6&1+6&1+6&1&1+6&1+6&1+6+6+6\\
\hline \end{tabular} \end{center}
\end{table}
\end{proof}    
Combining lemmas 3.2 and 3.3, lemma 3.1 is thus proved, and our proof for theorem 1.1 is complete.
\begin{remark}
It looks likely that this method can also be generalized to deal with the case when $A$ is a nonintegral matrix. In this case one has to formulate the problem a bit more carefully and do all the above analysis with more patience. In other words, we expect the above method to provide a complete list of rational matrices $A$ such that all solutions to \eqref{N} are real. Lastly, it is conceivable to speculate that the method may also be generalized to higher rank cases: Again, our goal is to get some inequalities on matrix entries from Bezout's theorem, which hopefully are sharp enough to give an upper bound for matrix entries depending only on the rank, thus giving a finite set of possibilities for the matrix $A$ for each rank, such that all solutions to \eqref{N1} are real. Thanks to the work of A.G. Khovanskii, we already have an almost satisfactory substitute for lemma 2.4 in higher rank cases: corollary 7 on page 80 of the book \cite{Fewnomials}. It provides an upper bound for the number of nondegenerate real solutions to \eqref{N1} depending only on the rank. (Although in rank 2 case, this upper bound is too large compared to ours in lemma 2.4. But the most important thing is the existence of such an upper bound in general.) Consequently, if we can find a good way to estimate the sum of multiplicities of solutions to certain homogeneous equations corresponding to \eqref{N1}, which are not solutions to \eqref{N1}, thus getting a substitute for lemma 2.3, we may achieve our goal. (provided that we also know how to deal with the exceptional cases when some of the real solutions to \eqref{N1} are degenerate) 
\end{remark}      
\section*{Acknowledgments}
We thank David Eisenbud, Richard Borcherds, Christian Zickert, and Junjie Zhou for discussions, and we are grateful to the organizers for the March 2010 AIM workshop on Mock theta functions, where the authors first learned about Nahm's conjecture from Werner Nahm and Don Zagier. We appreciate the referee's sugguestions, which were helpful for improving the presentation of the paper.

\end{document}